\documentclass[publicdomain]{eptcs}
\providecommand{\event}{AFL 2014}

\usepackage{amsfonts}
\usepackage{latexsym}
\usepackage{amsmath}

\input{tcilatex}
\newcommand{\eop}{{\nopagebreak\vspace{1ex}\nopagebreak
\hspace*{\fill}\mbox{$\Box$}}}
\renewenvironment{proof}{{\noindent\bf Proof.} }{\eop}

\begin{document}

\title{Cooperating Distributed Grammar Systems of Finite Index Working in
Hybrid Modes}
\author{Henning Fernau
\institute{Fachbereich 4---Abteilung Informatik\\
Universit{\"a}t Trier\\
D-54286 Trier, Germany}
\email{fernau@uni-trier.de}
\and Rudolf Freund 
\institute{Institut f\"{u}r Computersprachen\\
Technische Universit\"{a}t Wien\\
Favoritenstr. 9, A-1040 Wien, Austria}
\email{rudi@emcc.at}
\and Markus Holzer
\institute{Institut f\"ur Informatik\\
Universit\"at Gie\ss{}en,\\
  Arndtstra\ss{}e~2, D-35392 Gie\ss{}en, Germany}
\email{holzer@informatik.uni-giessen.de}
}

\def\titlerunning{Cooperating Distributed Grammar Systems of Finite Index 
Working in Hybrid Modes} 
\def\authorrunning{Henning Fernau, Rudolf Freund, and Markus Holzer}
\maketitle

\begin{abstract}
We study cooperating distributed grammar systems working in hybrid modes in
connection with the finite index restriction in two different ways: firstly,
we investigate cooperating distributed grammar systems working in hybrid
modes which characterize programmed grammars with the finite index
restriction; looking at the number of components of such systems, we obtain
surprisingly rich lattice structures for the inclusion relations between the
corresponding language families. Secondly, we impose the finite index
restriction on cooperating distributed grammar systems working in hybrid
modes themselves, which leads us to new characterizations of programmed
grammars of finite index.
\end{abstract}


\underline{Keywords:} CD grammar systems; finite index; hybrid modes;
programmed grammars

\underline{AMS MSC[2010] classification:} 68Q42; 68Q45

\section{Introduction}

\label{sec-intro}

Cooperating distributed (CD) grammar systems first were introduced in~\cite%
{MeeRoz78} with motivations related to two-level grammars. Later, the
investigation of CD grammar systems became a vivid area of research after
relating CD grammar systems with Artificial Intelligence (AI) notions~\cite%
{CsuDas90}, such as multi-agent systems or blackboard models for problem
solving. From this point of view, motivations for CD grammar systems can be
summarized as follows: several grammars (agents or experts in the framework
of AI), mainly consisting of rule sets (corresponding to scripts the agents
have to obey to) are cooperating in order to work on a sentential form
(representing their common work), finally generating terminal words (in this
way solving the problem). The picture one has in mind is that of several
grammars (mostly, these are simply classical context-free grammars called
\textquotedblleft components\textquotedblright\ in the theory of CD grammar
systems) \textquotedblleft sitting\textquotedblright\ around a table where
there is lying the common workpiece, a sentential form. Some component takes
this sentential form, works on it, i.e., it performs some derivation steps,
and then returns it onto the table such that another component may continue
the work.

In classical CD grammar systems, all components work in the same derivation
mode. It is of course natural to alleviate this requirement, because it
simply refers to different capabilities and working regulations of different
experts in the original CD motivation. This leads to the notion of so-called
hybrid CD grammar systems introduced by Mitrana and P\u{a}un in \cite%
{Mit93,Pau94}. We investigate internally hybrid derivation modes which
partly allow for new characterizations of the external hybridizations
explained above. This paper belongs to a series of papers on hybrid modes in
CD grammar systems: as predecessors, we mention that \cite{FerFre96}
introduces hybrid modes in CD array grammar systems as a natural
specification tool for array languages and \cite{FerHol96b}~investigates
accepting CD grammar systems with hybrid modes; the two most relevant papers
are \cite{FerFreHol01,FerHolFre03} where the most important aspects of
internal and external mode hybridizations are discussed for the case of word
languages.

Here, we will continue this line of research, focussing on the finite index
restriction. The paper is organized as follows. In the next section, we
introduce the necessary notions. In Section~\ref{sec-finite-index}, we
review important notions and results in connection with the finite index
restriction. Section~\ref{sec-counting-components} is devoted to the study
of internally hybrid CD grammar systems with the (explicit) restriction of
being of finite index; we establish infinite hierarchies with respect to the
number of components and the number of maximal derivation steps per
component. In Section~\ref{sec-CDGS-finite-index}, we refine our previous
analysis (published in \cite{FerHolFre03}) showing characterizations of
programmed grammars of finite index by several variants of (internally)
hybrid CD grammar systems, also considering the number of grammar components
as an additional descriptional complexity parameter. In the last section, we
review our results again and give a prospect on possible future work.

\section{Definitions}

\label{defs}

We assume the reader to be familiar with some basic notions of formal
language theory and regulated rewriting, as contained in \cite%
{RozenbergSalomaa1997} and \cite{DasPau89}. In particular, details on
programmed grammars can be found there. In general, we have the following
conventions: $\subseteq $ denotes inclusion, while~$\subset $ denotes strict
inclusion; the set of positive integers is denoted by ${\mathbb{N}}$. The
empty word is denoted by~$\lambda $; $|\alpha |_{A}$ denotes the number of
occurrences of the symbol $A$ in~$\alpha $. We consider two languages $%
L_{1},L_{2}$ to be equal if and only if $L_{1}\setminus \{\lambda
\}=L_{2}\setminus \{\lambda \}$, and we simply write $L_{1}=L_{2}$ in this
case. The families of languages generated by linear context-free and
context-free grammars are denoted by $\mathcal{L}(\mathrm{LIN})$ and $%
\mathcal{L}(\mathrm{CF})$, respectively, and the family of finite languages
is denoted by $\mathcal{L}(\mathrm{FIN})$. We attach $-\lambda $ in our
notations for formal language classes if erasing rules are not permitted.
Notice that we use bracket notations in order to express that the equation
holds both in case of forbidding erasing rules and in the case of admitting
erasing rules (consistently neglecting the contents between the brackets).

Next we introduce programmed grammars, a well-known concept in the area of
regulated rewriting.

A \emph{programmed grammar\/} is a septuple $G=(N,T,P,S,\Lambda ,\sigma
,\phi )$, where~$N$, $T$, and $S\in N$ are the set of nonterminals, the set
of terminals, and the start symbol, respectively. In the following we use~$%
V_{G} $ to denote the set $N\cup T$, which is the complete working alphabet
of the grammar. $P$ is the finite set of context-free rules $A\rightarrow z$
with $A\in N$ and $z\in V_{G}^{\ast }$, and~$\Lambda $ is a finite set of
labels (for the rules in~$P$), such that~$\Lambda $ can also be interpreted
as a function which outputs a rule when being given a label; $\sigma $ and $%
\phi $ are functions from~$\Lambda $ into the set of subsets of~$\Lambda $.
For $(x,r_{1})$, $(y,r_{2})$ in $V_{G}^{\ast }\times \Lambda $ and $\Lambda
(r_{1})=(A\rightarrow z)$, we write $(x,r_{1})\Rightarrow (y,r_{2})$ if and
only if either

\begin{enumerate}
\item $x=x_{1}Ax_{2}$, $y=x_{1}zx_{2}$, and $r_{2}\in \sigma (r_{1})$, or

\item $x=y$, the rule $A\to z$ is not applicable to~$x$, and $r_{2}\in \phi
(r_{1})$.
\end{enumerate}

In the latter case, the derivation step is performed in the so-called \emph{%
\ appearance checking mode}. The set $\sigma (r_{1})$ is called success
field and the set $\phi (r_{1})$ is called failure field of~$r_{1}$. As
usual, the reflexive transitive closure of~$\Rightarrow $ is denoted by~$%
\Longrightarrow ^{\ast }$. The language generated by~$G$ is defined as 
\begin{equation*}
L(G)=\{\,w\in T^{\ast }\mid (S,r_{1})\Longrightarrow ^{\ast }(w,r_{2})\text{
for some }r_{1},r_{2}\in \Lambda \}.
\end{equation*}%
The family of languages generated by [$\lambda $-free] programmed grammars
containing only context-free rules is denoted by $\mathcal{L}(\mathrm{P},%
\mathrm{CF}[-\lambda ],\mathrm{ac})$. When no appearance checking features
are involved, i.e., $\phi (r)=\emptyset $ for each label $r\in \Lambda $, we
obtain the family $\mathcal{L}(\mathrm{P},\mathrm{CF}[-\lambda ])$.

Finally, we now define cooperating distributed (CD) and hybrid cooperating
distributed (HCD) grammar systems.

A \emph{CD grammar system\/} of degree~$n$, with $n\geq 1$, is an $(n+3)$%
-tuple $G=(N,T,S,P_{1},P_{2},\dots ,P_{n})$, where~$N$, $T$~are disjoint
alphabets of nonterminal and terminal symbols, respectively, $S\in N$ is the
start symbol, and $P_{1},\ldots ,P_{n}$ are finite sets of rewriting rules
over $N\cup T$. Throughout this paper, we consider only regular, linear
context-free, and context-free rewriting rules. For $x,y\in (N\cup T)^{\ast
} $ and $1\leq i\leq n$, we write $x\Longrightarrow _{i}y$ if and only if $%
x=x_{1}Ax_{2}$, $y=x_{1}zx_{2}$ for some $A\rightarrow z\in P_{i}$. Hence,
subscript~$i$ refers to the component to be used. Accordingly, $%
x\Longrightarrow _{i}^{m}y$ denotes an $m$-step derivation using component
number~$i$, where $x\Longrightarrow _{i}^{0}y$ if and only if $x=y$.

We define the \emph{classical basic modes\/} $B=\{\,\ast ,t\,\}\cup \{\,\leq
k,=k,\geq k\mid k\in {\mathbb{N}}\,\}$ and let 
\begin{equation*}
D=B\cup \{\,(\geq k\wedge \leq \ell )\mid k,\ell \in {\mathbb{N}},k\leq \ell
\,\}\cup \{\,(t\wedge \leq k),(t\wedge =k),(t\wedge \geq k)\mid k\in {%
\mathbb{N}}\,\}.
\end{equation*}%
For $f\in D$ we define the relation $\Longrightarrow _{i}^{f}$ by 
\begin{equation*}
x\Longrightarrow _{i}^{f}y\iff \exists m\geq 0:(x\Longrightarrow
_{i}^{m}y\wedge P(f,m,i,y)),
\end{equation*}%
where~$P$ is a predicate defined as follows (let $k\in {\mathbb{N}}$ and $%
f_{1},f_{2}\in B$): 
\begin{equation*}
\begin{array}{l|c}
\text{predicate} & \text{definition} \\ \hline\hline
P(=k,m,i,y) & m=k \\ 
P(\leq k,m,i,y) & m\leq k \\ 
P(\geq k,m,i,y) & m\geq k \\ 
P(\ast ,m,i,y) & m\geq 0 \\ 
P(t,m,i,y) & \lnot \exists z(y\Longrightarrow _{i}z) \\ 
P((f_{1}\wedge f_{2}),m,i,y)\,\, & \,\,P(f_{1},m,i,y)\wedge P(f_{2},m,i,y)%
\end{array}%
\end{equation*}

Observe that not every combination of modes as introduced above is a
genuinely hybrid mode. For example, the $(\geq k\land \leq k)$-mode is just
another notation for the $=k$-mode. Especially, $*$ may be used as a ``don't
care'' in our subsequent notations, since $P((*\land f_2),m,i,y)$ if and
only if $P(f_2,m,i,y)$.

If each component of a CD grammar system may work in a different mode, then
we get the notion of an \emph{(externally) hybrid CD (HCD) grammar system\/}
of degree~$n$, with $n\geq 1$, which is an $(n+3)$-tuple $%
G=(N,T,S,(P_{1},f_{1}),(P_{2},f_{2}),\dots ,(P_{n},f_{n}))$, where $%
N,T,S,P_{1},\dots ,P_{n}$ are as in a CD grammar system, and $f_{i}\in D$,
for $1\leq i\leq n$. Thus, we can define the language \emph{\ generated\/}
by a HCD grammar system as:%
\begin{equation*}
\begin{array}[b]{llll}
L(G) & := & \{\,w\in T^{\ast }\mid & S\Rightarrow
_{i_{1}}^{f_{i_{1}}}w_{1}\Rightarrow _{i_{2}}^{f_{i_{2}}}\ldots \Rightarrow
_{i_{m-1}}^{f_{i_{m-1}}}w_{m-1}\Rightarrow _{i_{m}}^{f_{i_{m}}}w_{m}=w \\ 
&  &  & \text{with }m\geq 1\text{, }1\leq i_{j}\leq n\text{, and }1\leq
j\leq m\,\}%
\end{array}%
\end{equation*}

If $F\subseteq D$ and $X\in \{\mathrm{LIN},\mathrm{CF}\}$, then the family
of languages generated by [$\lambda $-free] HCD grammar systems with degree
at most~$n$ using rules of type~$X$, each component working in one of the
modes contained in $F$, is denoted by $\mathcal{L}(\mathrm{HCD}%
_{n},X[-\lambda ],F)$. In a similar way, we write $\mathcal{L}(\mathrm{HCD}%
_{\infty },X[-\lambda ],F)$ when the number of components is not restricted.
If~$F$ is a singleton~$\{f\}$, we simply write $\mathcal{L}(\mathrm{CD}%
_{n},X[-\lambda ],f)$, where $n\in {\mathbb{N}}\cup \{\infty \}$;
additionally, we write~$L_{f}(G)$ instead of~$L(G)$ to denote the language
generated by the CD grammar system~$G$ in the mode $f$.

The following example is taken from \cite[Theorem~24]{FerFreHol01}, as we
need this language in the following of this paper. This should also help to
clarify our definitions.

\begin{example}
\label{exa-hybrid1} The non-context-free language $L=\{\,a_{1}^{n}a_{2}^{n}%
\ldots a_{k+1}^{n}\mid n\geq 1\,\}$ can be generated by the CD grammar
system $G=(N,T,S_{1},P_{1},P_{2})$, where $P_{1},P_{2}$ work in the $%
(t\wedge \geq k)$-mode, $k\geq 2$. For both components, we take $%
N=\{\,S_{i},A_{i},A_{i}^{\prime },\mid 1\leq i\leq k\,\}$ as nonterminal
alphabet and $T=\{a_{1},\dots ,a_{k+1}\}$ as terminal alphabet. The
components $P_{1}$ and $P_{2}$ are defined as follows: 
\begin{eqnarray*}
P_{1} &=&\{\,S_{i}\rightarrow S_{i+1}\mid 1\leq i<k\,\}\cup
\{S_{k}\rightarrow A_{1}\cdots A_{k}\}\cup  \\
&&\{\,A_{i}^{\prime }\rightarrow A_{i}\mid 1\leq i\leq k\,\}\quad \text{and}
\\
P_{2} &=&\{\,A_{i}\rightarrow a_{i}A_{i}^{\prime }\mid 1\leq i\leq
k-1\,\}\cup \{A_{k}\rightarrow a_{k}A_{k}^{\prime }a_{k+1}\}\cup  \\
&&\{\,A_{i}\rightarrow a_{i}\mid 1\leq i\leq k-1\,\}\cup \{A_{k}\rightarrow
a_{k}a_{k+1}\}.
\end{eqnarray*}%
Then we have $L(G)=L$, since every derivation of~$G$ leading to a terminal
word is of the form 
\begin{equation*}
S_{1}\Longrightarrow _{1}^{=k}A_{1}\ldots A_{k}\hspace{0.3cm}\cdots \hspace{%
0.3cm}\Longrightarrow _{2}^{=k}a_{1}^{n}\ldots a_{k}^{n}a_{k+1}^{n},
\end{equation*}%
where the intermediate steps are of the form 
\begin{equation*}
a_{1}^{i}A_{1}\ldots a_{k}^{i}A_{k}a_{k+1}^{i}\Longrightarrow
_{2}^{=k}a_{1}^{i+1}A_{1}^{\prime i+1}{}_{k}A_{k}^{\prime
i+1}{}_{k+1}\Longrightarrow _{1}^{=k}a_{1}^{i+1}A_{1}\ldots
a_{k}^{i+1}A_{k}a_{k+1}^{i+1};
\end{equation*}%
if a non-vanishing number of occurrences of~$A_{i}^{\prime }$ less than~$k$
is obtained by using $P_{2}$ then neither~$P_{1}$ nor~$P_{2}$ can perform~$k$
derivation steps any more. Hence, $G$ generates~$L$.

The same grammar system, viewed as a $(\mathrm{CD}_{2},\mathrm{CF},(t\wedge
=k))$ grammar system, generates~$L$, too.
\end{example}

\section{The Finite Index Restriction}

\label{sec-finite-index}

The finite index restriction is defined as follows: let~$G$ be an arbitrary
grammar type (from those discussed in Section~\ref{defs}) and let~$N$, $T$,
and~$S\in N$ be its nonterminal alphabet, terminal alphabet, and axiom,
respectively. For a derivation 
\begin{equation*}
D:S=w_{1}\Longrightarrow w_{2}\Longrightarrow \cdots \Longrightarrow
w_{n}=w\in T^{\ast }
\end{equation*}%
according to~$G$, we set $\mathit{ind\/}(D,G)=\max \{\,|w_{i}|_{N}\mid 1\leq
i\leq n\,\}$. In the case of programmed grammars we assume to have a
derivation of the form 
\begin{equation*}
D:(S,r_{1})=(w_{1},r_{1})\Longrightarrow (w_{2},r_{2})\Longrightarrow \cdots
\Longrightarrow (w_{n},r_{n})=(w,r_{n})\in T^{\ast }\times \Lambda .
\end{equation*}%
For $w\in T^{\ast }$, we define $\mathit{ind\/}(w,G)=\min \{\,\mathit{ind\/}%
(D,G)\mid D$ is a derivation for $w$ in $G\,\}.$ The \emph{index of grammar~$%
G$\/} is defined as $\mathit{ind\/}(G)=\sup \{\,\mathit{ind\/}(w,G)\mid w\in
L(G)\,\}.$ For a language~$L$ in the family $\mathcal{L}(\mathrm{X})$ of
languages generated by grammars of type~$\mathrm{X}$, we define $\mathit{%
ind\/}_{\mathrm{X}}(L)=\inf \{\,ind(G)\mid L(G)=L$ and $G$ is of type $%
\mathrm{X}\,\}.$ For a family $\mathcal{L}(\mathrm{X})$, we set 
\begin{eqnarray*}
\mathcal{L}_{n}(\mathrm{X}) &=&\{\,L\mid L\in \mathcal{L}(\mathrm{X})\text{
and }ind\text{\/}_{\mathrm{X}}(L)\leq n\,\}\quad \text{for }n\in \mathbb{N}%
\text{, and} \\
\mathcal{L}_{\mathit{fin\/}}(\mathrm{X}) &=&\bigcup_{n\geq 1}\mathcal{L}_{n}(%
\mathrm{X})\text{.}
\end{eqnarray*}

It is well-known that the class of programmed languages of index~$m$ can be
characterized in various ways, compare, e.g., \cite%
{DasPau89,FerHol97a,RozVer78c}. Especially, normal forms are available. For
the reader's convenience, we quote \cite[Theorem~9]{FerHolFre03} in the
following, since we will use it to give a sharpened and broadened version of 
\cite[Theorem~3.26]{Csuetal94all}, which leads us to new characterizations
of the classes $\mathcal{L}_{m}(\mathrm{P},\mathrm{CF})$ and $\mathcal{L}_{%
\mathit{fin\/}}(\mathrm{P},\mathrm{CF})$.

\begin{theorem}
\label{thm:nsf} For every $(\mathrm{P},\mathrm{CF},\mathrm{ac})$ grammar $%
G=(N,T,P,S,\Lambda ,\sigma ,\phi )$ whose generated language is of index $%
n\in {\mathbb{N}}$, there exists an equivalent $(\mathrm{P},\mathrm{CF},%
\mathrm{ac})$ grammar $G^{\prime }=(N^{\prime },T,P^{\prime },S^{\prime
},\Lambda ,\sigma ^{\prime },\phi ^{\prime })$ whose generated language is
also of index~$n$ and which satisfies the following three properties:

\begin{enumerate}
\item There exists a special start production with a unique label $p_0$,
which is the only production where the start symbol $S^{\prime }$ appears.

\item There exists a function $f:\Lambda ^{\prime }\rightarrow {\mathbb{N}}%
_{0}^{N^{\prime }}$ such that, if $S^{\prime }\Longrightarrow ^{\ast
}v\Longrightarrow _{p}w$ is a derivation in $G^{\prime }$, then $%
(f(p))(A)=|v|_{A}$ for every nonterminal $A$.

\item If $D:S^{\prime }=v_{0}\Longrightarrow _{r_{1}}v_{1}\Longrightarrow
_{r_{2}}v_{2}\cdots \Longrightarrow _{r_{m}}v_{m}=w$ is a derivation in~$%
G^{\prime }$ then, for every $v_{i}$, $0\leq i\leq m$, and every nonterminal 
$A$, $|v_{i}|_{A}\leq 1$. In other words, every nonterminal occurs at most
once in any derivable sentential form.
\end{enumerate}

Moreover, we may assume that either $G^{\prime }$ is a $(\mathrm{P},\mathrm{%
CF})$ grammar, i.e., we have $\phi ^{\prime }=\emptyset $, or that $%
G^{\prime }$ is a $(\mathrm{P},\mathrm{CF},\mathrm{ut})$ grammar, i.e., we
have $\phi ^{\prime }=\sigma ^{\prime }$.
\end{theorem}

In the following, we will refer to a grammar satisfying the three conditions
listed above as \textit{nonterminal separation form (NSF)}.

Theorem 1 shows that, in contrast to the general case, where $\mathcal{L}(%
\mathrm{P},\mathrm{CF},\mathrm{ac})\supset \mathcal{L}(\mathrm{P},\mathrm{CF}%
)$, the appearance checking feature does not increase the generative power
of programmed grammars if the finite index restriction is imposed;
especially we have $\mathcal{L}_{m}(\mathrm{P},\mathrm{CF}[-\lambda ],%
\mathrm{ac})=\mathcal{L}_{m}(\mathrm{P},\mathrm{CF}[-\lambda ])$.

Recall that we have shown in \cite[Theorem~30]{FerHolFre03} the following
link between hybrid CDGS and the finite index restriction on programmed
grammars.

\begin{theorem}
\label{the-fin.ind.charac} Let $\ell \in {\mathbb{N}}$ and $\Delta \in
\{\leq ,=\}$. Then we have: 
\begin{eqnarray*}
\mathcal{L}(\mathrm{HCD}_{\infty },\mathrm{CF}[-\lambda ],\{\,(t\wedge
\Delta k)\mid k\geq 1\,\}) &=&\bigcup_{k\in {\mathbb{N}}}\mathcal{L}(\mathrm{%
CD}_{\infty },\mathrm{CF}[-\lambda ],(t\wedge \Delta k)) \\
&=&\mathcal{L}(\mathrm{CD}_{\infty },\mathrm{CF}[-\lambda ],(t\wedge \Delta
l)) \\
&=&\mathcal{L}(\mathrm{CD}_{\infty },\mathrm{CF}[-\lambda ],(t\wedge \Delta
1)) \\
&=&\mathcal{L}_{\mathit{fin\/}}(\mathrm{P},\mathrm{CF}[-\lambda ],\mathrm{ac}%
).
\end{eqnarray*}
\end{theorem}

Unfortunately, our proof did not bound the number of components of the CD
grammar system. This is not just a coincidence, as we will see in this paper.

\section{Infinite Hierarchies for CD Grammar Systems Working in Hybrid Modes 
\label{sec-counting-components}}

Our task will be the study of the language families $\mathcal{L}(\mathrm{CD}%
_{n},\mathrm{CF}[-\lambda ],(t\wedge \Delta k))$ for different $n,k\in {%
\mathbb{N}}$ and $\Delta \in \{\leq ,=\}$. First we give some
characterizations of well-known language families, namely the family of
finite languages and the family of linear languages.

\begin{lemma}
For every $k\in {\mathbb{N}}$, and $\Delta \in \{\leq ,=\}$, we have 
\begin{equation*}
\mathcal{L}(\mathrm{FIN})=\mathcal{L}(\mathrm{CD}_{1},\mathrm{CF}[-\lambda
],(t\wedge \Delta k)).
\end{equation*}
\end{lemma}

\begin{proof}
Since we have only one component, by definition of the $(t\wedge \Delta k)$%
-mode, every derivation has length at most~$k$, so that we only get finite
languages. If $L=\{w_{1},w_{2},\dots ,w_{m}\}\subseteq T^{\ast }$ is some
finite language, then the grammar $G=(\{S\}\times \{1,\dots ,k\},T,(S,1),P)$
with 
\begin{equation*}
P=\{\,(S,i)\rightarrow (S,i+1)\mid 1\leq i<k\,\}\cup \{\,(S,k)\rightarrow
w_{j}\mid 1\leq j\leq m\,\})
\end{equation*}%
generates~$L$.
\end{proof}

Now we turn our attention to CD grammar systems with two components working
in the $(t\wedge \Delta 1)$-mode for $\Delta \in \{\leq ,=\}$.

\begin{lemma}
For $\Delta \in \{\leq ,=\}$ we have $\mathcal{L}(\mathrm{LIN})=\mathcal{L}(%
\mathrm{CD}_{2},\mathrm{CF}[-\lambda ],(t\wedge \Delta 1)).$
\end{lemma}

\begin{proof}
Let~$L$ be generated by the linear grammar $G=(N,T,S,P)$. Grammar~$G$ is
simulated by the CD grammar system $G^{\prime }=(N\cup N^{\prime
},T,S,P_{1},P_{2})$ where~$N^{\prime }$ contains primed versions of the
nonterminals of~$G$, set~$P_{1}$ contains colouring unit productions $%
B\rightarrow B^{\prime }$ for every nonterminal $B\in N$, and~$P_{2}$
contains, for every production $A\rightarrow w\in P$, a production $%
A^{\prime }\rightarrow w$. The simulation of~$G$ by~$G^{\prime }$ proceeds
by applying~$P_{2}$ and~$P_{1}$ in sequence until the derivation stops.

On the other hand, it is easy to see that no sentential form generated by
some $(\mathrm{CD}_{2},\mathrm{CF}[-\lambda ],(t\wedge \Delta 1))$-system
(eventually leading to a terminal string) can contain more than one
nonterminal. Otherwise, we must have applied a production $A\rightarrow w$
of say the first component, where $w$ contains at least two nonterminals.
All nonterminals occurring in~$w$ cannot be processed further by the first
component, since otherwise it violates the $(t\wedge \Delta 1)$-mode
restriction. But nearly the same argument applies to the second component,
too: it can only process at most one of the nonterminals just introduced.
Hence, no terminal string is derivable in this way.

Therefore, one can omit all productions containing more than one nonterminal
on their right-hand sides, so that there are only linear rules left.
Furthermore, one can also omit all productions in a component containing a
nonterminal as its right-hand side which occurs also as the left-hand side
of the originally given component as this would lead to more than one
derivation step in the same component. Now, one can put all remaining
productions together yielding the rule set of a simulating linear grammar.
\end{proof}

In the general case, i.e., two components working together in the $(t\wedge
\Delta k)$-mode, for $k\in {\mathbb{N}}$ and $\Delta \in \{\leq ,=\}$, we
first give some lower bounds.

\begin{theorem}
\label{the-fin.hier1} Let $k\in {\mathbb{N}}$, $\Delta \in \{\leq ,=\}$.
Then we have:

\begin{enumerate}
\item $\mathcal{L}(\mathrm{LIN})=\mathcal{L}_{1}(\mathrm{CF})=\mathcal{L}(%
\mathrm{CD}_{2},\mathrm{CF}[-\lambda ],(t\wedge \Delta 1))$; 

\item $\mathcal{L}_{k}(\mathrm{CF})\subseteq \mathcal{L}(\mathrm{CD}_{2},%
\mathrm{CF}[-\lambda ],(t\wedge \Delta k))$ for $k\geq 1$, and

\item $\mathcal{L}_{k}(\mathrm{CF})\subset \mathcal{L}(\mathrm{CD}_{2},%
\mathrm{CF}[-\lambda ],(t\wedge =k))$ for $k>1$;

\item $\mathcal{L}_{\mathit{fin\/}}(\mathrm{CF})\subset \bigcup_{k\in {%
\mathbb{N}}}\mathcal{L}(\mathrm{CD}_{2},\mathrm{CF}[-\lambda ],(t\wedge =k))$%
.
\end{enumerate}
\end{theorem}

\begin{proof}

\begin{enumerate}
\item It is easy to see that $\mathcal{L}(\mathrm{LIN})=\mathcal{L}_{1}(%
\mathrm{CF})$. Hence, this statement is equivalent to the assertion of the
previous lemma.

\item Let $G=(N,T,S,P)$ be a context-free grammar of index~$k$. Without loss
of generality, we assume that every nonterminal occurs as the left-hand side
of some production in~$P$. Let $N^{\prime }$ be the set of primed
nonterminal symbols. Grammar~$G$ is simulated by the CD grammar system $%
G^{\prime }=(N\cup N^{\prime },T,S,P_{1},P_{2})$, where~$P_{1}$ contains
colouring unit productions $B\rightarrow B^{\prime }$, and $B\rightarrow B$
for every nonterminal $B\in N$, and~$P_{2}$, for every production $%
A\rightarrow w\in P$, contains productions $A^{\prime }\rightarrow w$ and $%
A^{\prime }\rightarrow A^{\prime }$. The unit productions $B\rightarrow B$
in $P_{1}$ and $A^{\prime }\rightarrow A^{\prime }$ in $P_{2}$ guarantee
that at most $k$ nonterminals can occur in any sentential form that can be
derived in $G^{\prime }$.

\item A separating example was already explained in Example~\ref{exa-hybrid1}%
: there the languages $\{\,a_{1}^{n}a_{2}^{n}\ldots a_{k+1}^{n}\mid n\geq
1\,\}$ was shown to be in $\mathcal{L}(\mathrm{CD}_{2},\mathrm{CF}[-\lambda
],(t\wedge =k))$ for $k\geq 2$, but obviously these languages are not
context-free.

\item Follows from~3.
\end{enumerate}
\end{proof}

Unfortunately, we do not know whether the inclusion 
\begin{equation*}
\mathcal{L}_{k}(\mathrm{CF})\subseteq \mathcal{L}(\mathrm{CD}_{2},\mathrm{CF}%
[-\lambda ],(t\wedge \leq k))
\end{equation*}%
in the previous theorem is strict or not. By the prolongation technique
introduced in \cite{FerFreHol01}, we know that the classes $\mathcal{L}(%
\mathrm{CD}_{n},\mathrm{CF}[-\lambda ],(t\wedge \Delta k))$, for $\Delta \in
\{\leq ,=\}$ form a prime number lattice, i.e., 
\begin{equation*}
\mathcal{L}(\mathrm{CD}_{n},\mathrm{CF}[-\lambda ],(t\wedge \Delta
k))\subseteq \mathcal{L}(\mathrm{CD}_{n},\mathrm{CF}[-\lambda ],(t\wedge
\Delta \ell \cdot k)\quad \text{for }\ell \in \mathbb{N},
\end{equation*}%
with the least element $\mathcal{L}(\mathrm{CD}_{n},\mathrm{CF}[-\lambda
],(t\wedge \Delta 1))$. This prolongation technique is based on the simple
idea to \textquotedblleft slow down\textquotedblright\ a derivation using $%
A\rightarrow w$ of the original CDGS by intercalating productions of the
form $A\rightarrow A^{\prime }$, $A^{\prime }\rightarrow A^{\prime \prime }$%
, \dots , $A^{(j)}\rightarrow w$ within the simulating CDGS. It will be used
on several occasions in this paper. Obviously, we also have the trivial
inclusions 
\begin{equation*}
\mathcal{L}(\mathrm{CD}_{n},\mathrm{CF}[-\lambda ],(t\wedge \Delta
k))\subseteq \mathcal{L}(\mathrm{CD}_{n+1},\mathrm{CF}[-\lambda ],(t\wedge
\Delta k))\quad \text{for }\Delta \in \{\leq ,=\}.
\end{equation*}

The question arises whether all these hierarchies are strict. At least we
will be able to show that both with respect to $k$ -- for a fixed number of
components $n$ -- as well as with respect to the number of components $n$ --
for a fixed derivation mode $(t\wedge \Delta k)$, $\Delta \in \{\leq ,=\}$
-- we obtain infinite hiearachies. In order to prove these hierarchies, we
show some general theorems relating the number of components and the bound
of the number of symbols to be rewritten by one component with the finite
index of a simulating programmed grammar.

\begin{theorem}
\label{inclusion}Let $n,k\in {\mathbb{N}}$ and $\Delta \in \{\leq ,=\}$.
Then, we have 
\begin{equation*}
\mathcal{L}(\mathrm{CD}_{n},\mathrm{CF}[-\lambda ],(t\wedge \Delta
k))\subseteq \mathcal{L}_{n\cdot k}(\mathrm{P},\mathrm{CF}[-\lambda ],ac).
\end{equation*}
\end{theorem}

\begin{proof}
Let $G=(N,T,S,P_{1},P_{2},\dots ,P_{n})$ be a CD grammar system working in
the $(t\wedge =k)$-mode. Let $P_{i}=\{\,A_{ij}\rightarrow w_{ij}\mid 1\leq
j\leq N(i)\,\}$. $G$ can be simulated by the programmed grammar $G^{\prime
}=(N\cup \{F\},T,S,P,\Lambda ,\sigma ,\phi )$, with label-set 
\begin{eqnarray*}
\Lambda  &=&\{\,(i,j,\kappa )\mid 1\leq i\leq n,1\leq j\leq N(i),1\leq
\kappa \leq k\,\} \\
&\cup &\{\,(i,j)\mid 1\leq i\leq n,1\leq j\leq N(i)\,\}.
\end{eqnarray*}

Now, let $1\leq i\leq n,1\leq j\leq N(i),1\leq \kappa \leq k$. Then, we set $%
\Lambda ((i,j,\kappa ))=A_{ij}\rightarrow w_{ij}$, success field $\sigma
((i,j,\kappa ))=\{\,(i,j^{\prime },\kappa +1)\mid 1\leq j^{\prime }\leq
N(i)\,\}$, if $\kappa <k$, and $\sigma ((i,j,k))=\{(i,1)\}$, and failure
field $\phi ((i,j,\kappa ))=\emptyset $. Moreover, we let $\Lambda
((i,j))=A_{ij}\rightarrow F$, failure field $\phi ((i,j))=\{(i,j+1)\}$ if $%
j<N(i)$, $\phi ((i,N(i)))=\{\,(i^{\prime },j^{\prime },1)\mid 1\leq
i^{\prime }\leq n,1\leq j^{\prime }\leq N(i^{\prime })\,\}$, and success
field $\sigma ((i,j))=\emptyset $.

An application of~$P_{i}$ is simulated by a sequence of productions labeled
with 
\begin{equation*}
(i,j_{1},1),\ldots ,(i,j_{k},k),(i,1),\ldots ,(i,N(i)).
\end{equation*}%
In each such sequence, at most~$k$ symbols can be processed. Since there are~%
$n$ sets of productions $P_{i}$, only sentential forms containing at most $%
n\cdot k$ nonterminals can hope for termination. Therefore, the simulating
programmed grammar has index at most~$n\cdot k$, which can be seen by
induction.

The $(t\wedge \leq k)$-mode case can be treated in a similar way: we just
define $\phi ((i,j,\kappa ))$ to equal $\sigma ((i,j,\kappa ))$ instead of
taking $\phi ((i,j,\kappa ))=\emptyset $.
\end{proof}

Before we can establish the infinite hierarchies for the families $\mathcal{L%
}(\mathrm{CD}_{n},\mathrm{CF}[-\lambda ],(t\wedge \Delta k))$ with respect
to~$n$ and~$k$, respectively, we need the following theorem shown in~\cite[%
page 160, Theorem 3.1.7]{DasPau89}:

\begin{theorem}
\label{seplan}$S_{n+1}=\{\,b(a^{i}b)^{2\cdot (n+1)}\mid i\geq 1\,\}\in 
\mathcal{L}_{n+1}(\mathrm{P},\mathrm{CF})\setminus \mathcal{L}_{n}(\mathrm{P}%
,\mathrm{CF})$ for all $n\in \mathbb{N}$.
\end{theorem}

These separating languages can also be generated by CD grammar systems
working in the internally hybrid modes $(t\wedge =k+1)$ and $(t\wedge \leq
k+1)$:

\begin{theorem}
\label{SnkinCD} Let $n,k\in {\mathbb{N}}$.

\begin{enumerate}
\item $S_{n\cdot k}\in \mathcal{L}(\mathrm{CD}_{n+1},\mathrm{CF}-\lambda
,(t\wedge =k+1))$, i.e., $S_{n\cdot k}$ can be generated by a CD-grammar
system with $n+1$ context-free components, without erasing productions,
working in the $(t\wedge =k+1)$-mode.

\item $S_{n\cdot k}\in \mathcal{L}(\mathrm{CD}_{2n+1},\mathrm{CF}-\lambda
,(t\wedge \leq k+1))$, i.e., $S_{n\cdot k}$ can be generated by CD-grammar
system with $2\cdot n+1$ context-free components, without erasing
productions, working in the $(t\wedge \leq k+1)$-mode.
\end{enumerate}
\end{theorem}

\begin{proof}
We first construct a CD grammar system 
\begin{equation*}
G=(N,\{a,b\},(S,0),P_{0},P_{1,1}\cup P_{1,2},\ldots ,P_{n,1}\cup P_{n,2})
\end{equation*}%
working in the $(t\wedge =k+1)$-mode generating language~$S_{n\cdot k}$. Let 
\begin{eqnarray*}
N &=&\{\,(S,i)\mid 0\leq i\leq n\,\} \\
&\cup &\{\,(q_{i},0),(q_{i},1)\mid 1\leq i\leq n\,\} \\
&\cup &\{\,(t_{i},j),(t_{i}^{\prime },j)\mid 1\leq i\leq n,0\leq j\leq k\,\}
\\
&\cup &\{\,(A_{i},0),(A_{i},1)\mid 1\leq i\leq k\cdot n\,\}
\end{eqnarray*}%
and let the set of productions be as follows: 
\begin{eqnarray*}
P_{0} &=&\{(S,0)\rightarrow (S,1),(S,k)\rightarrow
(q_{1},0)(A_{1},0)b(A_{2},0)b\ldots b(A_{n\cdot k},0)b\} \\
&\cup &\{\,(S,i)\rightarrow (S,i+1)\mid 1\leq i<k\,\} \\
&\cup &\{\,(A_{i},1)\rightarrow (A_{i},0)\mid 1\leq i\leq n\cdot k\,\} \\
&\cup &\{\,(q_{i},1)\rightarrow (q_{i+1},0)\mid 1\leq i<n\,\}\cup
\{(q_{n},1)\rightarrow (q_{1},0),(q_{n},1)\rightarrow (t_{1},0)\} \\
&\cup &\{\,(t_{i}^{\prime },j)\rightarrow (t_{i}^{\prime },j+1)\mid 1\leq
i\leq n,0\leq j<k\,\} \\
&\cup &\{\,(t_{i}^{\prime },k)\rightarrow (t_{i+1},0)\mid 1\leq i<n\,\}\cup
\{(t_{n}^{\prime },k)\rightarrow b\}.
\end{eqnarray*}%
For every $1\leq i\leq n$ let $P_{i}=P_{i,1}\cup P_{i,2}$ where 
\begin{eqnarray*}
P_{i,1} &=&\{(q_{i},0)\rightarrow (q_{i},1)\} \\
&\cup &\{\,(A_{j},0)\rightarrow a(A_{j},1)a\mid (i-1)\cdot k\leq j\leq
i\cdot k\,\}\quad \text{and} \\
P_{i,2} &=&\{(t_{i},0)\rightarrow (t_{i}^{\prime },0)\} \\
&\cup &\{\,(A_{j},0)\rightarrow aba\mid (i-1)\cdot k\leq j\leq i\cdot k\,\}.
\end{eqnarray*}

The only way to start the derivation is to use~$P_{0}$ obtaining the word 
\begin{equation*}
(q_{1},0)(A_{1},0)b(A_{2},0)b\ldots b(A_{n\cdot k},0)b.
\end{equation*}%
To continue the derivation one always has to apply~$P_{i}$ and~$P_{0}$ in
sequence: $P_{i}$ is only successfully applicable to a sentential form
beginning with a letter $(q_{i},0)$ or $(t_{i},0)$, and $P_{0}$ is only
successfully applicable to a sentential form beginning with a letter $%
(q_{i},1)$ or $(t_{i}^{\prime },0)$.

Assume we have a sentential form starting with letter $(q_{i},0)$, then
rules from~$P_{i}$ replace exactly~$k$ occurrences of nonterminals $%
(A_{j},0) $, for $(i-1)\cdot k\leq j\leq i\cdot k$, by $a(A_{j},1)a$ or $aba$%
, respectively, and the label is changed to $(q_{i},1)$. Now applying the
corresponding rules from $P_{0}$ (the only way to continue), the derivation
would block if at least one of the symbols $(A_{j},0)$ from the previous
step were replaced by $aba$. This ensures that all previously used
productions are non-terminating. In case the sentential from starts with a
letter $(t_{i},0)$, then an application of~$P_{i}$ followed by~$P_{0}$
checks whether the terminating rules form~$P_{i}$ were all used or not.

Thus, starting with the word $(q_{1},0)(A_{1},0)b(A_{2},0)b\ldots
b(A_{n\cdot k},0)b$, one cycle, i.e., an application of $%
P_{1},P_{0},P_{2},P_{0},\ldots ,P_{n}$, and $P_{0}$ leads to 
\begin{equation*}
(q_{1},0)a(A_{1},0)aba(A_{2},0)ab\ldots ba(A_{n\cdot k},0)ab\,,
\end{equation*}%
and in general, running the cycle~$\ell $ times, to the word 
\begin{equation*}
(q_{1},0)a^{\ell }(A_{1},0)a^{\ell }ba^{\ell }(A_{2},0)a^{\ell }b\ldots
ba^{\ell }(A_{n\cdot k},0)a^{\ell }b.
\end{equation*}%
Note that in the last application of~$P_{0}$ we could have taken the rule $%
(q_{n},1)\rightarrow (t_{1},0)$ in order to terminate the derivation process
after having finished the next cycle. After that cycle the grammar~$G$ has
generated the word $b(a^{\ell +1}b)^{2n\cdot k}$.

In case of the $(t\wedge \leq k+1)$-mode, we use the same construction as
above, but now treat each $P_{i,j}$, for $1\leq i\leq n$ and $1\leq j\leq 2$%
, as an independent component of the grammar. This gives the bound $2\cdot
n+1$ on the number of components.
\end{proof}

Now we are ready to investigate the families $\mathcal{L}(\mathrm{CD}_{n},%
\mathrm{CF}[-\lambda ],(t\wedge \Delta k))$ in more detail. First, let the
number of components $n$ be fixed.

\begin{corollary}
Let $\Delta \in \{\leq ,=\}$ and $n\in {\mathbb{N}}$ be fixed. The hierarchy
of the families of languages $\mathcal{L}_{k}:=\mathcal{L}(\mathrm{CD}_{n},%
\mathrm{CF}[-\lambda ],(t\wedge \Delta k))$ with respect to~$k\in {\mathbb{N}%
}$ is infinite, i.e., for every $k\in {\mathbb{N}}$ there exists an $m\in {%
\mathbb{N}}$, $m>k$, such that $\mathcal{L}_{k}\subset \mathcal{L}_{m}$.
\end{corollary}

\begin{proof}
We first consider the $(t\wedge =k)$-mode. By the preceding theorem we have
that 
\begin{equation*}
S_{n\cdot m}\in \mathcal{L}(\mathrm{CD}_{n+1},\mathrm{CF}[-\lambda
],(t\wedge ={m+1})).
\end{equation*}
On the other hand, $S_{n\cdot m}\notin \mathcal{L}_{n\cdot m-1}(\mathrm{P},%
\mathrm{CF})$ according to Theorem~\ref{seplan} and, because of 
\begin{equation*}
\mathcal{L}(\mathrm{CD}_{n+1},\mathrm{CF}[-\lambda ],(t\wedge ={k}%
))\subseteq \mathcal{L}_{(n+1)\cdot k}(\mathrm{P},\mathrm{CF})
\end{equation*}
according to Theorem~\ref{inclusion}, $S_{n\cdot m}$ therefore cannot belong
to $\mathcal{L}(\mathrm{CD}_{n+1},\mathrm{CF}[-\lambda ],(t\wedge ={k}))$,
provided $n\cdot m-1\geq (n+1)\cdot k$, i.e., $m\geq \frac{n+1}{n}k+\frac{1}{%
n}$.

For the $(t\wedge \leq k)$-mode, we can argue in a similar way: By the
preceding theorem we have that 
\begin{equation*}
S_{n\cdot m}\in \mathcal{L}(\mathrm{CD}_{2n+1},\mathrm{CF}[-\lambda
],(t\wedge \leq {m+1}))
\end{equation*}%
and therefore $S_{n\cdot m}\in \mathcal{L}(\mathrm{CD}_{2n+2},\mathrm{CF}%
[-\lambda ],(t\wedge \leq {m+1}))$, too. On the other hand,~$S_{n\cdot
m}\notin \mathcal{L}_{n\cdot m-1}(\mathrm{P},\mathrm{CF})$ according to
Theorem~\ref{seplan} and, because of 
\begin{equation*}
\mathcal{L}(\mathrm{CD}_{2n+1},\mathrm{CF}[-\lambda ],(t\wedge \leq {k}%
))\subseteq \mathcal{L}_{(2n+1)\cdot k}(\mathrm{P},\mathrm{CF})\subseteq 
\mathcal{L}_{(2n+2)\cdot k}(\mathrm{P},\mathrm{CF})
\end{equation*}%
according to Theorem~\ref{inclusion}, $S_{n\cdot m}$ therefore cannot belong
to $\mathcal{L}(\mathrm{CD}_{2n+1},\mathrm{CF}[-\lambda ],(t\wedge \leq {k}%
))\cap \mathcal{L}(\mathrm{CD}_{2n+2},\mathrm{CF}[-\lambda ],(t\wedge \leq {k%
}))$, provided $n\cdot m-1\geq (2n+2)\cdot k$, i.e., $m\geq 2\frac{n+1}{n}k+%
\frac{1}{n}$.
\end{proof}

We now consider the other hierarchy, i.e., we fix~$k$ and vary the number of
components:

\begin{corollary}
Let $\Delta \in \{\leq ,=\}$ and $k\in {\mathbb{N}}$ be fixed. The hierarchy
of the families of languages $\mathcal{L}_{n}:=\mathcal{L}(\mathrm{CD}_{n},%
\mathrm{CF}[-\lambda ],(t\wedge \Delta k))$ with respect to $n\in {\mathbb{N}%
}$ is infinite, i.e., for every $n\in {\mathbb{N}}$ there exists an $m\in {%
\mathbb{N}}$, $m>n$, such that $\mathcal{L}_{n}\subset \mathcal{L}_{m}$.
\end{corollary}

\begin{proof}
We argue in a similar way as in the preceding corollary. First consider the $%
(t\wedge =k)$-mode of derivation. We already know that%
\begin{equation*}
S_{m\cdot k}\in \mathcal{L}(\mathrm{CD}_{m+1},\mathrm{CF}[-\lambda
],(t\wedge =k{+1})).
\end{equation*}
according to Theorem~\ref{seplan}, but $S_{m\cdot k}\notin \mathcal{L}%
_{m\cdot k-1}(\mathrm{P},\mathrm{CF})$ according to Theorem~\ref{seplan}
and, because of 
\begin{equation*}
\mathcal{L}(\mathrm{CD}_{n},\mathrm{CF}[-\lambda ],(t\wedge ={k+1}%
))\subseteq \mathcal{L}_{n\cdot (k+1)}(\mathrm{P},\mathrm{CF})
\end{equation*}%
according to Theorem~\ref{inclusion}, $S_{m\cdot k}$ therefore cannot belong
to $\mathcal{L}(\mathrm{CD}_{n},\mathrm{CF}[-\lambda ],(t\wedge ={k+1}))$,
provided $m\cdot k-1\geq n\cdot (k+1)$, i.e., $m\geq \frac{k+1}{k}n+\frac{1}{%
k}$.

By a similar reasoning, in case of the $(t\wedge \leq k)$-mode, for 
\begin{equation*}
S_{m\cdot k}\in \mathcal{L}(\mathrm{CD}_{2m+1},\mathrm{CF}[-\lambda
],(t\wedge \leq k{+1}))\subseteq \mathcal{L}(\mathrm{CD}_{2m+2},\mathrm{CF}%
[-\lambda ],(t\wedge \leq k{+1}))
\end{equation*}%
we obtain $S_{m\cdot k}\notin \mathcal{L}(\mathrm{CD}_{2n+1},\mathrm{CF}%
[-\lambda ],(t\wedge ={k+1}))\cap \mathcal{L}(\mathrm{CD}_{2n+2},\mathrm{CF}%
[-\lambda ],(t\wedge ={k+1}))$, provided $m\cdot k-1\geq (2n+2)\cdot (k+1)$,
i.e., $m\geq 2\frac{k+1}{k}n+\frac{2(k+1)}{k}$ ($=2(1+\frac{1}{k})n+2(1+%
\frac{1}{k})$).

In both cases, we see that the hierarchy with respect to the number of
components is infinite.
\end{proof}

Finally, let us consider the hierarchies for the \textquotedblleft small
cases\textquotedblright\ of~$n$.

\begin{lemma}
\label{lem-cdn} Let $k\in {\mathbb{N}}$ and $\Delta \in \{\leq ,=\}$. 
\begin{eqnarray*}
\mathcal{L}(\mathrm{CD}_{1},\mathrm{CF}[-\lambda ],(t\wedge \Delta k))
&\subset &\mathcal{L}(\mathrm{CD}_{2},\mathrm{CF}[-\lambda ],(t\wedge \Delta
1)) \\
&\subset &\mathcal{L}(\mathrm{CD}_{3},\mathrm{CF}[-\lambda ],(t\wedge \Delta
k)).
\end{eqnarray*}
\end{lemma}

\begin{proof}
By our previous considerations, we know that $\mathcal{L}(\mathrm{CD}_{1},%
\mathrm{CF}[-\lambda ],(t\wedge \Delta k))$ and $\mathcal{L}(\mathrm{CD}_{2},%
\mathrm{CF}[-\lambda ],(t\wedge \Delta 1))$ coincide with $\mathcal{L}(%
\mathrm{FIN})$ and $\mathcal{L}(\mathrm{LIN})$, respectively, which already
proves the first strict inclusion.

Now consider the non-linear language $\{\,a^{n}b^{n}a^{m}b^{m}\mid n,m\in {%
\mathbb{N}}\,\}$, which is generated by a CD grammar system 
\begin{equation*}
G=(\{S,A,B,A^{\prime },B^{\prime }\},\{a,b\},S,P_{1},P_{2},P_{3})
\end{equation*}%
taking $k=1$ with the following three components: 
\begin{eqnarray*}
P_{1} &=&\{S\rightarrow AB,A^{\prime }\rightarrow A,B^{\prime }\rightarrow
B\} \\
P_{2} &=&\{A\rightarrow aA^{\prime }b,A\rightarrow ab,B^{\prime }\rightarrow
B^{\prime }\} \\
P_{3} &=&\{B\rightarrow aB^{\prime }b,B\rightarrow ab,A\rightarrow
A,A^{\prime }\rightarrow A^{\prime }\}.
\end{eqnarray*}%
First, $P_{1}$ and $P_{2}$ have to be applied in sequence, say $n$ times,
until $P_{2}$ uses the rule $A\rightarrow ab$. Now, $P_{3}$ can be applied.
Then, $P_{1}$ and $P_{3}$ must be applied in sequence, say $m-1$ times,
until $P_{3}$ terminates the whole derivation using $B\rightarrow ab$. In
this way, a word $a^{n}b^{n}a^{m}b^{m}$ is derived. By the prolongation
technique, the claimed assertion follows for $k>1$.
\end{proof}

We conclude this section by remarking that the results presented in this
section (originally contained in the Technical Report~\cite{FerFreHol96})
have been employed to show the following theorem in~\cite{BorHol99} that
nicely complements our results here; we state these below with the notations
of our paper.

\begin{theorem}
\label{thm-borhol99} Let $n,k\geq 1$. Then we have

\begin{enumerate}
\item $\mathcal{L}(\mathrm{CD}_{n},\mathrm{CF}[-\lambda ],(t\wedge
=k))\subset \mathcal{L}(\mathrm{CD}_{n+2},\mathrm{CF}[-\lambda ],(t\wedge
=k+1))$ and

\item $\mathcal{L}(\mathrm{CD}_{n},\mathrm{CF}[-\lambda ],(t\wedge \leq
k))\subset \mathcal{L}(\mathrm{CD}_{2\cdot (n+1)},\mathrm{CF}[-\lambda
],(t\wedge \leq k+1))$.
\end{enumerate}
\end{theorem}

\section{CD Grammar Systems and Programmed Grammars of Finite Index}

\label{sec-CDGS-finite-index}

In this section we consider the finite index property for CD and HCD grammar
systems and how they relate to programmed grammars of finite index in more
detail.

\begin{theorem}
Let $m\in {\mathbb{N}}$, $FI=\{t\}\cup \{\,=m^{\prime \prime },\geq
m^{\prime \prime },(\geq m^{\prime \prime }\wedge \leq m^{\prime }),(t\wedge
=k)$,\ $(t\wedge \leq k),(t\wedge \geq k)\mid m^{\prime },m^{\prime \prime
},k\in {\mathbb{N}},m^{\prime }\geq m^{\prime \prime }\geq m\,\}$, and~$F$
contain all the hybrid modes considered in this paper, i.e., $F=\{\,(t\wedge
=k)$,\ $(t\wedge \leq k)\mid k\in {\mathbb{N}}\,\}$. Let $f\in FI$. Then%
\begin{eqnarray*}
\mathcal{L}_{m}(\mathrm{P},\mathrm{CF}[-\lambda ]) &=&\mathcal{L}_{m}(%
\mathrm{CD}_{\infty },\mathrm{CF}[-\lambda ],f) \\
&=&\mathcal{L}_{m}(\mathrm{HCD}_{\infty },\mathrm{CF}[-\lambda ],F).
\end{eqnarray*}
\end{theorem}

\begin{proof}
Looking through all the proofs showing the containment of HCD languages
within programmed languages with appearance checking, it is easily seen that
all these constructions preserve a finite index restriction.

Hence, it only remains to show that 
\begin{equation*}
\mathcal{L}_{m}(\mathrm{P},\mathrm{CF}[-\lambda ])\subseteq \mathcal{L}_{m}(%
\mathrm{CD}_{\infty },\mathrm{CF}[-\lambda ],f)
\end{equation*}%
for every $f\in FI$. Let $L\in \mathcal{L}_{m}(\mathrm{P},\mathrm{CF})$ be
generated by a programmed grammar $G=(N,T,P,S,\Lambda ,\sigma )$ in NSF.
Especially, there exists a function $f:\Lambda \rightarrow {\mathbb{N}}%
_{0}^{N}$ such that, if $S%
\mathrel{\mathop{        \hbox{$\Rightarrow$}}\limits^{\scriptstyle *}
    \limits_{\scriptstyle }}v%
\mathrel{\mathop{        \hbox{$\Rightarrow$}}\limits^{\scriptstyle }
    \limits_{\scriptstyle p}}w$ is a derivation in $G$, then $%
(f(p))(A)=|v|_{A}\leq 1$ for every nonterminal~$A$.

We construct a simulating CD grammar system $G^{\prime }$ given by 
\begin{equation*}
((N\times\Lambda )\cup(\{\,i\in{\mathbb{N}}\mid i\leq m\,\}\times N),
T,(1,S),\{P_I\}\cup\{\,P_{p,q}\mid p\in\Lambda \land q\in\sigma(p)\,\} )
\end{equation*}
of index~$m$ all of whose components are working in one of the modes $=m$, $%
\geq m$, $t$, $(\geq m,\leq m^{\prime })$ (with $m^{\prime }\geq m$), $%
(t\land=m)$, $(t\land\leq m)$, or $(t\land \geq m)$. Consider a production $%
\Lambda (p)=A\to w $ of~$G$. We can assume (check) that $(f(p))(A)=1$.
Furthermore, define 
\begin{equation*}
n:=\sum_{B\in N}(f(p))(B)\leq m
\end{equation*}
is the number of nonterminals in the current string (within a possible
derivation leading to an application of~$p$).

Let the homomorphism $h_{p,q}:(N\times\{p\}\cup T)^*\to (N\times\{q\}\cup
T)^*$ be defined by $(A,p)\mapsto (A,q)$ for $A\in N$, $a\mapsto a$ for $%
a\in T$.

For every $q\in\sigma$, we introduce a component $P_{p,q}$ within the CD
grammar system containing the following productions:

If $n=m$, then $(A,p)\rightarrow h_{p,q}(w)$ simulates the (successful)
application of rule~$p$. If $n<m$, then we prolong the derivation in the
following way: 
\begin{equation*}
(A,p)\rightarrow (1,A),\quad (1,A)\rightarrow (2,A),\quad \dots ,\quad
(m-n,A)\rightarrow h_{p,q}(w).
\end{equation*}%
$(B,p)\rightarrow h_{p,q}(B)$ for $B\in N\setminus \{A\}$ keeps track of the
information of the current state.

As initialization component, we take $P_I$ containing 
\begin{equation*}
(1,S)\to (2,S),\quad\dots,\quad (m,S)\to (S,p)
\end{equation*}
for every $p\in\Lambda $ such that $(f(p))(S)=1$.

Observe that, by induction, to every sentential form derivable from the
initial symbol $(1,S)$, any component applied to it can make either $0$
steps (so, we selected the wrong one) or exactly $m$ steps.

By a simple prolongation trick, we can also take components working in one
of the modes $=m^{\prime \prime }$, $\geq m^{\prime \prime }$, $(\geq
m^{\prime \prime },\leq m^{\prime })$ (with $m^{\prime }\geq m^{\prime
\prime }$), for some $m^{\prime \prime }\geq m$.

Since the construction given in Theorem~\ref{the-fin.ind.charac} is
index-preserving, we can also take arbitrary $(t\land =k)$ or $(t\land \leq
k)$ components instead of requiring $k\geq m$. Since the $t$-mode and the $%
(t\land \geq 1)$-mode are identical, the prolongation technique delivers the
result for the $(t\land\geq k)$-mode for $k\in{\mathbb{N}}$ in general.
\end{proof}

Our theorem readily implies a characterization of programmed languages of
general finite index. We summarize this fact together with results obtained 
\textit{via} a different simulation in \cite[Theorem~3.26]{Csuetal94all} in
the following corollary.

\begin{corollary}
Let $m,k,k^{\prime }\in {\mathbb{N}}$, $k\geq 2$, and $\Delta \in \{\leq
,=,\geq \}$. Then following families of languages coincide with $\mathcal{L}%
_{\mathit{fin\/}}(\mathrm{P},\mathrm{CF}[-\lambda ])$:

\begin{enumerate}
\item $\mathcal{L}_{\mathit{fin\/}}(\mathrm{CD}_{\infty },\mathrm{CF}%
[-\lambda ],t)$,

\item $\mathcal{L}_{\mathit{fin\/}}(\mathrm{CD}_{\infty },\mathrm{CF}%
[-\lambda ],=k)$,

\item $\mathcal{L}_{\mathit{fin\/}}(\mathrm{CD}_{\infty },\mathrm{CF}%
[-\lambda ],\geq k)$,

\item $\mathcal{L}(\mathrm{CD}_{\infty },\mathrm{CF}[-\lambda ],(t\wedge
=k)) $,

\item $\mathcal{L}(\mathrm{CD}_{\infty },\mathrm{CF}[-\lambda ],(t\wedge
\leq k))$, and

\item $\mathcal{L}_{\mathit{fin\/}}(\mathrm{CD}_{\infty },\mathrm{CF}%
[-\lambda ],(t\wedge \Delta k^{\prime }))$.
\end{enumerate}
\end{corollary}

As regards the number of components, we can state the following:

\begin{theorem}
Let $m\in {\mathbb{N}}$, $\Delta \in \{\leq ,=,\geq \}$. Then we have

\begin{enumerate}
\item $\mathcal{L}_{m}(\mathrm{CD}_{\infty },\mathrm{CF}[-\lambda ],\Delta
m)=\mathcal{L}_{m}(\mathrm{CD}_{3},\mathrm{CF}[-\lambda ],\Delta m)$ and

\item $\mathcal{L}_{m}(\mathrm{CD}_{\infty },\mathrm{CF}[-\lambda ],(t\wedge
\Delta m))=\mathcal{L}_{m}(\mathrm{CD}_{3},\mathrm{CF}[-\lambda ],(t\wedge
\Delta m))$.
\end{enumerate}
\end{theorem}

\begin{proof}
Since we restrict our attention to languages of index~$m$, we can simply
carry over the proofs of the type ``three is enough'' for $t$-mode
components, see \cite[Theorem~3.10]{Csuetal94all} and \cite[Lemma~2]{Mit93}.

In case of the $(t\wedge =m)$-mode, we have to go back to the simulation of
the programmed grammar given in the preceding theorem. It is clear that, due
to the nonterminal separation form (NSF) (see Theorem~\ref{thm:nsf}), we can
prolong the simulation of a single nonterminal symbol.
\end{proof}

It is quite natural to compare the families of languages defined by CD
grammar systems obtained \textit{via} the restriction of being of finite
index $m$ with their unrestricted counterparts. In our case, it is
interesting to see that also these unrestricted counterparts deliver
languages of finite index. However, as we have seen in Theorem~\ref{SnkinCD}%
, 
\begin{equation*}
S_{n\cdot k}\in \mathcal{L}(\mathrm{CD}_{n+1},\mathrm{CF}[-\lambda
],(t\wedge =k+1)).
\end{equation*}%
Especially, we have 
\begin{equation*}
S_{2\cdot (m-1)}\in \mathcal{L}(\mathrm{CD}_{3},\mathrm{CF}[-\lambda
],(t\wedge =m)).
\end{equation*}

Since $S_{2\cdot (m-1)}$ is not of (programmed) index $2m-3$, we can state:

\begin{corollary}
For $m\in {\mathbb{N}}$, we have 
\begin{equation*}
\mathcal{L}_{m}(\mathrm{CD}_{3},\mathrm{CF}[-\lambda ],(t\wedge =m))\subset 
\mathcal{L}(\mathrm{CD}_{3},\mathrm{CF}[-\lambda ],(t\wedge =m)).
\end{equation*}
\end{corollary}

\begin{proof}
Our previous considerations deliver the case $m>2$, since $S_{2(m-1)}$ is
not a (programmed) language of index $2m-3$, and therefore not a
(programmed) language of index $m$; hence, 
\begin{equation*}
S_{2\cdot (m-1)}\notin \mathcal{L}_{m}(\mathrm{CD}_{3},\mathrm{CF}[-\lambda
],(t\wedge =m)).
\end{equation*}%
In case $m=2$, we know that 
\begin{equation*}
S_{3}\notin \mathcal{L}_{2}(\mathrm{CD}_{3}, \mathrm{CF}[-\lambda ],(t\wedge
=2)).
\end{equation*}

On the other hand, the $(\mathrm{CD}_{3},\mathrm{CF}[-\lambda ],(t\wedge
=2)) $ grammar system 
\begin{equation*}
G=(\{S,A,B,A^{\prime },B^{\prime }C,C^{\prime },B^{\prime \prime
},F\},\{a,b\},S,P_{1},P_{2},P_{3})
\end{equation*}%
with the following three components generates $S_{3}$, starting with $S$: 
\begin{eqnarray*}
P_{1} &=&\{A\rightarrow aA^{\prime }a,A\rightarrow aba,B\rightarrow
aB^{\prime }a,B\rightarrow B^{\prime \prime }\} \\
P_{2} &=&\{B^{\prime }\rightarrow B,C\rightarrow aC^{\prime }a,A\rightarrow
F,B^{\prime \prime }\rightarrow aba,C\rightarrow aba\} \\
P_{3} &=&\{S\rightarrow bAbBbCb,C^{\prime }\rightarrow C,A^{\prime
}\rightarrow A,B^{\prime }\rightarrow F\}.
\end{eqnarray*}

Finally, we have $\mathcal{L}_{1}(\mathrm{CD}_{3},\mathrm{CF}[-\lambda ],
(t\wedge =1))=\mathcal{L}(\mathrm{LIN})$. The example $\{%
\,a^{n}b^{n}a^{m}b^{m}\mid n,m\in {\mathbb{N}}\,\}\notin \mathcal{L}(\mathrm{%
LIN})$ in Lemma~\ref{lem-cdn} was shown to be in $\mathcal{L}(\mathrm{CD}%
_{3},\mathrm{CF}[-\lambda ],(t\wedge =1))$.
\end{proof}

If we admit four components (or arbitrarily many), by a similar reasoning we
can separate all corresponding classes, since $3(m-1)=3m-3>m$ for $m\geq 2$.

\begin{corollary}
For $m,n\in {\mathbb{N}}$, $n\geq 4$ (or $n=\infty $), we have 
\begin{equation*}
\mathcal{L}_{m}(\mathrm{CD}_{n},\mathrm{CF}[-\lambda ],(t\wedge =m))\subset 
\mathcal{L}(\mathrm{CD}_{n},\mathrm{CF}[-\lambda ],(t\wedge =m)).
\end{equation*}
\end{corollary}

Observe that the results exhibited in the last two corollaries are quite
astonishing if one keeps in mind that 
\begin{eqnarray*}
\mathcal{L}_{\mathit{fin\/}}(\mathrm{P},\mathrm{CF}) &=&\bigcup_{m\in {%
\mathbb{N}}}\mathcal{L}_{m}(\mathrm{CD}_{n},\mathrm{CF}[-\lambda ],(t\wedge
=m)) \\
&=&\bigcup_{m\in {\mathbb{N}}}\mathcal{L}(\mathrm{CD}_{n},\mathrm{CF}%
[-\lambda ],(t\wedge =m))
\end{eqnarray*}%
for all $n\in {\mathbb{N}},n>2$.

\section{Conclusions and Prospects}

In this paper we have studied CD grammar systems working in the internally
hybrid modes $(t\wedge =m)$ and $(t\wedge \leq m)$ together with the finite
index restriction. Showing specific relations to programmed grammars of
finite index, we were able to establish infinite hierarchies for CD grammar
systems of finite index working in the internally hybrid modes $(t\wedge =k)$
and $(t\wedge \leq k)$ both with respect to the number of components $n$ and
the number of maximal steps $k$. However, many quite natural questions still
remain open. For instance, Theorem~\ref{thm-borhol99} leaves open the
strictness of several natural inclusion relations relating the parameters
\textquotedblleft numbers of components\textquotedblright\ $n$ and
\textquotedblleft step number bound\textquotedblright\ $k$.

It is well-known that ET0L systems are tightly related to CD grammar systems
working in the $t$-mode. In the literature, several step-bound restrictions
have been discussed in relation with parallel systems, see~\cite{Fer03a} for
an overview. Are these somehow related to the internally hybrid systems
discussed in this paper (and their companions)? Or do hybrid modes lead to
new (natural) derivation modes for parallel systems? In particular, the
finite index restriction studied in Section~\ref{sec-CDGS-finite-index}
could be of interest in this context, although (and also because) we are not
aware of a study of finite index in the context of limited parallel
rewriting, which might be an interesting research question in its own.

\paragraph{Acknowledgements}

Most of the research of the first and last author was undertaken while being
affiliated to Wilhelm-Schickard Institut f\"ur Informatik, Universit\"at
T\"ubingen, Sand~13, D-72076 T\"ubingen, Germany. Part of the research of
the first author was supported by Deutsche Forschungsgemeinschaft, grant DFG
La 618/3-1/2 ``Kom\-ple\-xit\"ats\-theo\-re\-ti\-sche Methoden f\"ur die
ad\"aquate Modellierung paralleler Be\-rech\-nun\-gen.''

\bibliographystyle{eptcs}
\bibliography{AFLCD}

\end{document}